\documentclass[runningheads]{llncs}
\usepackage{graphicx}
\usepackage{latexsym,amssymb}
\usepackage{amsmath,enumerate}
\usepackage{paralist}
\usepackage{algorithm}
\usepackage[noend]{algpseudocode}
\usepackage{xcolor}
\usepackage{multicol}
\usepackage{subcaption}
\usepackage{varwidth}
\usepackage{mathtools}
\usepackage{array}

\newtheorem{invariant}{Invariant}[section]

\renewenvironment{proof}{{\bf Proof:  }}{\hfill\rule{2mm}{2mm}\vspace*{5pt}}

\newcommand{\cancel}{\textsf{Cancel}}
\newcommand{\ins}{\textsf{Insert}}
\newcommand{\del}{\textsf{Delete}}
\newcommand{\create}{\textsf{Create}}
\newcommand{\query}{\textsf{Query}}
\newcommand{\now}{\textsf{now}}

\begin{document}

\title{Upper and Lower Bounds for Fully Retroactive Graph Problems\thanks{The full version of the paper can be found at \url{https://arxiv.org/abs/1910.03332}.}}
%
%
\author{Monika Henzinger\inst{1}
	\and
Xiaowei Wu\inst{2}
\thanks{This work was done in part when the author was a postdoc at University of Vienna.}}
\authorrunning{M. Henzinger and X. Wu}
%
\institute{Faculty of Computer Science, University of Vienna, Vienna, Austria\\
\email{monika.henzinger@univie.ac.at}
\and
IOTSC, University of Macau, Macau, China\\
\email{xiaoweiwu@um.edu.mo}
}
\maketitle           

\begin{abstract}
	Classic dynamic data structure problems maintain a data structure subject to a sequence $S$ of updates and they answer queries using the latest version of the data structure, i.e., the data structure after processing the whole sequence.
	To handle operations that change the sequence $S$ of updates, Demaine et al.~\cite{talg/DemaineIL07} introduced {\em retroactive data structures} (RDS).
	A retroactive operation modifies the update sequence $S$ in a given position $t$, called {\em time}, and either creates or cancels an update in $S$ at time $t$.
	A {\em fully retroactive} data structure supports queries at any time $t$: a query at time $t$ is answered using only the updates of $S$ up to time $t$.
	%
	While efficient RDS have been proposed for classic data structures, e.g., stack, priority queue and binary search tree, the retroactive version of graph problems are rarely studied.
	
	In this paper we study retroactive graph problems including connectivity, minimum spanning forest (MSF), maximum degree, etc.
	We show that under the OMv conjecture (proposed by Henzinger et al.~\cite{stoc/HenzingerKNS15}), there does not exist fully RDS maintaining connectivity or MSF, or incremental fully RDS maintaining the maximum degree with $O(n^{1-\epsilon})$ time per operation, for any constant $\epsilon > 0$.
	Furthermore, We provide RDS with \emph{almost tight} time per operation.
	We give fully RDS for maintaining the maximum degree, connectivity and MSF in $\tilde{O}(n)$ time per operation.
	We also give an algorithm for the incremental (insertion-only) fully retroactive connectivity with $\tilde{O}(1)$ time per operation, showing that the lower bound cannot be extended to this setting.
	
	
	We also study a restricted version of RDS, where the only change to $S$ is the swap of neighboring updates and show that for this problem we can beat the above hardness result. This also implies the first non-trivial dynamic Reeb graph computation algorithm.
\end{abstract}


\section{Introduction}

A dynamic data structure problem maintains a data structure on a set of elements subject to element insertions, deletions and modifications.
An efficient dynamic algorithm updates the data structure after each element update, and supports queries on the latest version of the data structure.
That is, an update can only append an operation to the end of the operation sequence, and a query can only be made on the data structure with all updates applied.
However, in some applications, we are interested in modifying the update sequence in the middle.
For example, if some past update on a database is mistaken and needs to be removed, we do not want to rollback the whole database by canceling all updates after the mistaken one.
Besides, in some scenarios we are interested in querying the data structure when only part of the updates are applied, e.g., to answer questions like ``which facebook user had the most friends in Jan 1st, 2015?''.
This motivates retroactive data structures (RDS) that were introduced by Demaine et al.~\cite{talg/DemaineIL07}.
They support (1) modifications to the historical sequence of updates performed on the data structure, and (2) queries on the data structure when only a prefix of the updates is applied.

Formally speaking, the data structure is defined by a sequence $S$ of updates, each of which is associated with a time $t$.
A RDS supports {operations} that create or cancel an {update} at any time $t$. 
There are $|S|+1$ versions of the data structure, on any of which a query can be made.
Throughout this paper, we use \emph{update} to denote a modification to the data structure, and \emph{operation} to denote a retroactive action that creates or cancels an update.
Depending on the queries supported, Demaine et al.~\cite{talg/DemaineIL07} defined two classes of RDS: a \emph{partially retroactive} data structure supports queries only at the present time, i.e., on the latest version of the data structure, while a \emph{fully retroactive} data structure supports queries on any version of the data structure.
For dynamic problems in which the ordering of updates is not important, e.g., maintaining a dictionary, standard dynamic algorithms are automatically partially retroactive.
However, maintaining a fully RDS can be much more difficult, as a retroactive operation at time $t$ can possibly change the outcome of \emph{all} queries after time $t$.
For example, an insertion of a very small key into a min-heap at time $t$ can possibly change the output of every find-min query after time $t$.
In general, there does not exist efficient transformation from partially RDS to fully retroactive ones.
Demaine et al.~\cite{talg/DemaineIL07} provided a general checkpointing method that converts a partially RDS into a fully retroactive one, with an $O(\sqrt{T})$ multiplicative overhead in the update and query time, where $T = |S|$.
Indeed, the $O(\sqrt{T})$ multiplicative overhead is shown to be tight for some data structures~\cite{swat/ChenDGWXY18}, under some well-known computational hardness conjectures.

\paragraph{Prior Works.}
Demaine et al.~\cite{talg/DemaineIL07} provided a partially retroactive priority queue with $O(\log T)$ update time and $O(1)$ query time, which implies a fully retroactive priority queue with $O(\sqrt{T}\log T)$ update and query time.
The result was later improved by Demaine et al.~\cite{wads/DemaineKLSY15}, who proposed a fully retroactive priority queue with amortized polylogarithmic update and query time.
They introduced a hierarchical checkpointing technique, which maintains a balanced binary tree with the set of updates as the leaves.
Giora and Kaplan~\cite{talg/GiyoraK09} considered the dynamic vertical ray shooting problem, and proposed a data structure that supports horizontal line segment insertions and deletions, and queries that report the first segment intersecting a vertical ray from a query point in worst case $O(\log T)$ time.
Their data structure implies a fully retroactive binary search tree with $O(\log T)$ update and query time. 

While dynamic graph problems flourished in the past decades, their retroactive versions are rarely studied.
Dynamic algorithms maintaining connectivity~\cite{jacm/HenzingerK99,jacm/HolmLT01}, minimum spanning forest (MSF)~\cite{jacm/HolmLT01,esa/HolmRW15} and maximal matching~\cite{sicomp/BaswanaGS18,focs/Solomon16,soda/BernsteinFH19} with polylogarithmic update and query time are known, but their fully retroactive versions have not been studied yet.
One exception is the empirical analysis of~\cite{de2020fully} on the fully retroactive minimum spanning tree (MST) problem.
For the aforementioned problems, the dynamic data structures are equivalent to the partially retroactive ones.
Thus by Demaine et al.'s reduction~\cite{talg/DemaineIL07}, there exist fully RDS for these problems, with $\tilde{O}(\sqrt{T})$ update and query time.\footnote{Throughout this paper we use $\tilde{O}()$ to hide the polylogarithmic factors in $T$ and $n$.}
Note that, in general, the number of updates $T$ can be much larger than the number of nodes and edges in the graph.
Roditty and Zwick~\cite{sicomp/RodittyZ16} proposed a fully RDS that supports queries of strong connectivity between two nodes at any version of the graph, subject to directed edge insertions and deletions.
However, the retroactive operations are restricted to be \emph{incremental}: each operation either creates an insertion of edge at the end of the update sequence, or cancels an existing update.
Their algorithm answers each query in worst case $O(1)$ time and handles each update in amortized $O(m\cdot \alpha(m, n))$ time, where $m$ is the number of edges in the graph and $\alpha()$ is the inverse Ackermann function~\cite{jacm/Tarjan75}. 
Chen et al.~\cite{swat/ChenDGWXY18} showed that there exist data structures for which a gap of $(\min\{n,\sqrt{T}\})^{1-o(1)}$ exists in the time per operation between partially and fully RDS, under some well-known conjectures.
However, these data structure are not graph data structures, but rather unusual data structures.

\paragraph{Our Results.}
We study the fully RDS for graph problems, providing for a variety of fundamental graph problems efficient incremental fully RDS and almost matching upper and lower bounds for their fully dynamic fully retroactive counterparts.
We start with some strong hardness results on the update and query time for fully RDS on several graph problems, assuming the online boolean matrix-vector multiplication (OMv) conjecture~\cite{stoc/HenzingerKNS15}.
Our hardness results show that for many of the problems we study in this paper, it is difficult to get RDS with truly sublinear time per operation.

\begin{theorem}\label{th:main_hardness}
	Assuming the OMv conjecture, there do not exist data structures for the following problems with $O(n^{1-\epsilon})$ update and query time subject to edge insertions/deletions:
	\begin{compactitem}
		\item fully retroactive connectivity, maximal matching, MSF, maximum density;
		\item incremental fully retroactive maximum degree.
	\end{compactitem}
\end{theorem}

Our hardness results hold even when the edges are unweighted.
For maintaining a maximal matching and spanning forest, we assume that queries are on the size of the matching and the forest, respectively.
In the full version~\cite{henzinger2019upper}, we show that the same hardness result holds for fully RDS supporting queries on the existence of perfect matching.
Moreover, some of our hardness results apply even to approximation algorithms.
For the graph problems we study in this paper (in which the ordering of updates is not important, such as connectivity, maximal matching, and MSF), the partially retroactive setting is the same as the standard dynamic setting and can, thus, be solved in polylogarithmic time.
Hence our hardness results imply a polynomial gap in the time per operation between the partially and fully RDS.
Our hard instances consist of a sequence of $T = \Theta(n^2)$ operations and queries.
Thus they also imply that under the OMv conjecture, getting an $O(T^{1/2-\epsilon})$ time per operation is impossible (for the aforementioned problems).
Under the combinatorial boolean matrix multiplication conjecture, we show that our hardness results hold even when all operations are given before any query is made (which we refer to as the \emph{offline} version of the problem), as long as the data structures are combinatorial.

\smallskip

We also provide RDS with almost tight time per operation.
We first consider the \emph{incremental} setting, in which a retroactive operation either creates an insertion, or cancels an existing insertion.
In other words, the creation of a deletion is not supported. 
We provide incremental fully RDS for maintaining connectivity and spanning forest (SF) with polylogarithmic update and query time.
Observe that the incremental partially retroactive setting is at least as hard as the (non-retroactive) fully dynamic setting, as the cancel operation in the retroactive setting serves the function of deletion in the dynamic setting.
Our data structure for maintaining connectivity and spanning forest supports only unweighted edge insertions and deletions.
However, we show that it can be extended to support weighted edge insertions and deletions, resulting in an $(1+\epsilon)$-approximation MSF with polylogarithmic update and query time.

\begin{theorem}\label{th:main_inc}
	There exist incremental fully RDS maintaining connectivity, spanning forest, and an $(1+\epsilon)$-approximation MSF with $\tilde{O}(1)$ amortized update time and $\tilde{O}(1)$ worst case query time.
\end{theorem}

Note that while the incremental connectivity problem is equivalent to the union-find problem in the dynamic setting, their retroactive versions are different, at least as defined by Demaine et al.~\cite{talg/DemaineIL07}.
In the retroactive setting, an insertion of an edge at time $t$ that connects two different connected components in the connectivity problem corresponds to a union operation between two equivalence classes in the union-find problem at time $t$. 
If we insert another edge connecting the same two components at time $t' > t$, then its corresponding operation in the union-find data structure of Demaine et al. is \emph{illegal}, as two equivalence classes can not be united twice (at time $t'$ and $t$).
In other words, the set of retroactive operations allowed for the two problems are different.
Consequently, the fully retroactive union-find data structure by Demaine et al.~\cite{talg/DemaineIL07} with $O(\log T)$ time per operation can not be used to achieve the above result.

\smallskip

We also present data structures maintaining MSF and the maximum degree that supports (creation of) insertions and deletions of weighted edges.
By Theorem~\ref{th:main_hardness}, our data structures have almost tight time per operation.

\begin{theorem}\label{th:main_fully}
	There exist fully RDS maintaining connectivity, MSF and maximum degree of an undirected graph with amortized $\tilde{O}(1)$ update time and worst case $O(n\log T)$ query time.
\end{theorem}

Our algorithmic results are obtained by maintaining a scapegoat tree~\cite{soda/GalperinR93} with $O(T)$ leaves, each of which is an interval defined by the times of two consecutive updates.\footnote{A similar data structure was mentioned in~\cite[Theorem 6]{talg/DemaineIL07}. However, they built a segment tree~\cite{bentley1977algorithms} on the leaves and some details on maintaining the tree were missing.}
Each internal node stores a set of edges, and maintains a data structure (depending on the problem) to support the queries.
The tree structure allows efficient retrieval of the edges that exist at time $t$ by examining $O(\log T)$ internal nodes.
Moreover, it can be shown that each edge is stored in $O(\log T)$ internal nodes.
Consequently, for problems that admit linear time algorithms, e.g., maximal matching, each query can be answered in worst case $O(m\log T)$ time, where $m$ is the maximum number of edges.
For maintaining connectivity, MSF and the maximum degree, we show that the query time can be improved to $O(n\log T)$, by maintaining a sparse data structure in each internal node of the scapegoat tree. 
%
A similar (yet different) data structure was used by Demaine et al.~\cite{wads/DemaineKLSY15} to maintain the set of retroactive operations sorted by time for their fully retroactive priority queue data structure.
In their \emph{checkpoint tree}, a scapegoat tree is maintained with the set of retroactive operations being the leaves.
Each internal node $u$ maintains a partially RDS induced by the operations (leaves) in the subtree rooted at $u$.
Consequently, if an element is stored at some node $u$, it is also stored at the parent of $u$.
In contrast, in our data structure, the set of elements stored at an internal node is disjoint from the set of the elements stored at its children.
Moreover, since we do not maintain partially RDS in internal nodes, we do not need to maintain explicitly the set of invalid operations, e.g., a deletion of an edge that is inserted by an operation in another subtree.
This property is crucial for efficient data structures on graph problems when edges are inserted and deleted multiple times.
We summarize our results in Table~\ref{table:results} as follows (where Retro. stands for Retroactive).

\vspace*{-10pt}

\begin{table*}[h]
	\footnotesize
	\begin{center}
		{
			\renewcommand{\arraystretch}{1.1}
			\begin{tabular}{ |c|c|c|c| }
				\hline
				& Incremental & Fully Retro. & Hardness \\ \hline
				Maximum Degree & $\mathbf{\tilde{O}(n)}$ & $\mathbf{\tilde{O}(n)}$ & $\mathbf{\Omega(n^{1-o(1)})}$ (Incremental) \\ \hline 
				Connectivity, SF & $\mathbf{\tilde{O}(1)}$ & $\mathbf{\tilde{O}(n)}$ & $\mathbf{\Omega(n^{1-o(1)})}$ (Fully Retro.) \\ \hline 
				MSF & $\tilde{O}(1)$, $(1+\epsilon)$-approx. & $\mathbf{\tilde{O}(n)}$ & $\mathbf{\Omega(n^{1-o(1)})}$ (Fully Retro.) \\ \hline
				Maximal Matching & $\tilde{O}(m)$ & $\tilde{O}(m)$ & $\Omega(n^{1-o(1)})$ (Fully Retro.)\\ \hline
			\end{tabular}
		}
	\end{center}
	\caption{Summary of results. The complexity in each cell is for the amortized time per operation. The results in bold are almost tight.}
	\label{table:results}
\end{table*}

\vspace*{-20pt}

As we will show in Section~\ref{sec:fully}, in the (classic) dynamic setting, there exists a simple data structure that maintains the maximum degree of an unweighted graph in worst case $O(1)$ time.
On the other hand, it is well-known that maintaining connectivity takes time $\Omega(\log n)$~\cite{sicomp/PatrascuD06}.
In other words, maintaining maximum degree is ``easier'' than maintaining connectivity in the dynamic setting.
However, Theorem~\ref{th:main_hardness} and Theorem~\ref{th:main_inc} imply that in the incremental fully retroactive setting this relationship is reversed: maintaining the maximum degree cannot be done in truly sublinear time under the OMv conjecture, while the connectivity problem can be solved in polylogarithmic time.
This interesting observation illustrates how different RDS can be, when compared to dynamic data structures.

\smallskip

Our study of RDS was motivated by an application in computational topology, specifically the problem of dynamically maintaining a Reeb graph~\cite{thesis/parsa2014algorithms}. However, for that problem a restricted version of the fully retroactive connectivity problem has to be solved. Specifically, no updates can be inserted or deleted in $S$, but \emph{the order of two neighboring updates can be reversed}.
We call such an operation a \emph{swap operation}.
Interestingly, under this restricted setting we can beat the lower bounds (Theorem~\ref{th:main_hardness}) for the general retroactive setting.
We give a $\tilde O(1)$ time data structure for this restricted version, leading to the first non-trivial dynamic Reeb graph algorithm.
Indeed, our approach can be extended to a general class of problems, for which the answer only depends on the currently existing ``elements'' and \emph{not} on the order of the updates.

\begin{theorem}\label{th:main-swap}
	Suppose for a dynamic version of a problem there exists a data structure with $T_u$ update time, $T_q$ query time, and space complexity $\mathcal{M}$.
	Then for any integer $1\leq \tau\leq T$ and any fixed $T$ updates $S$ (each of which is associated with a time), there exists a fully RDS for the problem supporting swap operations with $O(T_u)$ update time and $O( T_q + (\tau-1) \cdot T_u)$ query time.
	The data structure uses $O(T\cdot \mathcal{M}/\tau)$ space.
\end{theorem}


\paragraph{Other Related Work.}
Persistence~\cite{jcss/DriscollSST89,jal/FiatK03} is another concept of dynamic data structures that consider updates with times.
The data structures maintain (and support queries on) several versions of the data structure simultaneously.
Operations of a persistent data structure can be performed on any version of the data structure, which produces a new version.
A key difference between persistent data structure and retroactive ones is that a retroactive operation at time $t$ changes \emph{all} later versions of a RDS, while in a persistent one each version is considered an unchangeable archive.
Other efficient RDS, e.g., for dynamic point location and nearest neighbor search, can be found on~\cite{soda/Blelloch08,esa/DickersonEG10,isaac/GoodrichS11,thesis/parsa2014algorithms}.


\section{Preliminaries}\label{sec:prelim}

In a RDS, each update and query is associated with a time $t$, where $t$ is a real number. We use $\now = +\infty$ to denote the present time.
Each retroactive operation creates or cancels an update of the graph at time $t$, and each query at time $t$ reveals some property of the graph at time $t$.
Specifically, we use $\create(\text{update},t)$ to denote a retroactive operation that creates an update at time $t$ and $\cancel(t)$ to denote the retroactive operation that removes the update at time $t$.
In this paper, updates are edge insertions $\ins(e)$ and deletions $\del(e)$.
Moreover, we assume that all operations are \emph{legal}.
For example, $\create(\del(e),t)$ can only be issued when edge $e$ exists at time $t$ and is not deleted after time $t$; $\cancel(t)$ can only be issued when there is an update at time $t$.
We assume that the initial graph is empty, and all updates and queries take place at different times.

A fully RDS supports queries $\query(\text{parameters}, t)$ at any time $t$, where the set of parameters can be empty.
A query made at time $t$ should be answered on the version of the graph at time $t$, on which only updates up to time $t$ are applied.
For example, for the connectivity problem, $\query(u,v,t)$ answers whether $u$ and $v$ are connected by edges that exist at time $t$.

Throughout the whole paper, we use $n$ to denote the number of nodes (which is fixed).
We use $T$ to denote the current number of updates (which is dynamic), excluding the updates that are cancelled.
A RDS maintains a sequence of updates $S$ sorted in ascending order of time.
The size of $S$ is $T$, which increases by one after each $\create(\text{update},t)$, and decreases by one after each $\cancel(t)$.
The set $S$ defines $T+1$ versions of the graph, and a query can be made on any of them.
Note the difference between an operation and an update with the definition of $S$: $S$ is a set of updates that define the versions of the graph, while operations modify $S$.
Throughout this paper we assume that the word size of the RAM is $O(\log n)$, and $T$ is polynomial\footnote{Note that any data structure need to store the $|S| = T$ updates. Thus if $T$ is too large then the space complexity would be already unacceptable. Alternatively we can assume that the word size is $O(\log T)$ as the parameters in the operations might have size $\Theta(\log T)$.} in $n$.
Consequently, we have $O(\log T) = O(\log n)$ and we only need constant words to represent any time $t$.
We also assume that the weights of edges are polynomial in $n$.

\paragraph{Incremental Fully Retroactive.}
In the incremental case, the retroactive operation $\create(\del(e),t)$ does not exist, i.e., $S$ contains only insertions of edges (at different times).
Note that in the incremental case the $\cancel(t)$ operation can still be issued, which removes one update (insertion) from $S$.

\smallskip

As we will show later, for maintaining connectivity, the incremental case is substantially easier than the general case; while for maintaining the maximum degree, even the incremental case can be very difficult.
The following definition will be useful for our data structures.

\begin{definition}[Lifespan]
	For each edge $e$ inserted at time $t_a$ and whose earliest deletion after $t_a$ is at time $t_b$ (which is $\now$ if it is not deleted), let $L_e = (t_a,t_b]$ be the \emph{lifespan} of $e$.
\end{definition}

While an edge can be inserted and deleted multiple times, to ease our notation we regard $e$ as a new edge every time it is inserted.
By definition, the set of edges existing at time $t$ is given by $E_t = \{ e: t\in L_e \}$,.
A query made at time $t$ should be answered based on the graph $G_t := (V,E_t)$.

\section{Lower Bounds}\label{sec:lower_bounds}

We present the hardness result for maintaining fully retroactive connectivity based on the OMv conjecture in this section.
That is, we prove Theorem~\ref{th:main_hardness} for the fully retroactive connectivity problem.
The proofs of other hardness results are included in the full version of the paper~\cite{henzinger2019upper}.
We first show that for almost all graph problems, ``natural'' fully retroactive algorithms can not have update and query time $o(\log T)$.
Consider a simple fully RDS on a graph with $n=2$ nodes.
The data structure needs to support insertions and deletions of the edge between the two nodes, and queries of whether the edge exists at time $t$.
We show that the problem is at least as hard as searching a key among $T$ sorted elements.
Thus no comparison-based\footnote{Given a query at time $t$, a comparison-based algorithm compare $t$ with times of other updates to identify the one with time closest to $t$.} fully RDS has update and query time $o(\log T)$.

Let $k_1<k_2<\ldots<k_T$ be $T$ points in time.
For each $i=1,2,\ldots,T$, we insert an edge $e=(u,v)$ at time $t = k_i$ and delete the edge immediately.
In other words, the edge $e$ exists only at time $k_1,k_2,\ldots,k_T$.
Assume that you are given a query operation with time parameter $k$, to check whether $k$ is in $\{k_1,\ldots,k_T\}$, it suffices to query whether the edge exists at time $k$.
Given that any comparison-based search requires $\Omega(\log T)$ time to find an element, we have an $\Omega(\log T)$ lower bound on the query time, for comparison-based fully retroactive algorithms of a large class of dynamic graph problems (including maximum degree, connectivity, maximal matching, etc).
The following lemma justifies the $O(\log T)$ factor that appears in the time per operation of our data structures.

\begin{lemma}
	Any comparison-based fully retroactive algorithm has $\Omega(\log T)$ time per operation.
\end{lemma}

%

\paragraph{OMv Conjecture.}
In the Online Boolean Matrix-Vector Multiplication (OMv) problem, the algorithm is given an $n\times n$ boolean matrix $M$, while a sequence of $n$ length-$n$ boolean vectors $v_1,v_2,\ldots,v_n$ arrive online.
The algorithm needs to output the vector $Mv_i$ before seeing the next vector $v_{i+1}$.
The OMv conjecture~\cite{stoc/HenzingerKNS15} states that there does not exit algorithm with $O(n^{3-\epsilon})$ running time for this problem, for any constant $\epsilon > 0$.

\medskip

We give a reduction from the OMv problem to fully retroactive connectivity as follows.
The reductions to other graph problems are similar.
Given an instance of the OMv problem consisting of an $n\times n$ matrix $M$ and an online sequence of $n$-dimensional vectors $\{v_i\}_{i\in[n]}$, let $m_i$ be the $i$-th row of matrix $M$.
Let $|x|$ denote the number of non-zero entries in a vector $x$.
We construct a graph with $n+2$ nodes $a,b,u_1,\ldots,u_n$.
We describe and construct a sequence of retroactive operations from the OMv instance as follows.


Recall that we assume all operations have different time.
However, for convenience, we use the following description.
By saying that we construct a set of retroactive operations $S$ at time $t$, we fix an arbitrary order of the operations in $S$, and construct the operations one by one, at time $t, t+\epsilon,\ldots,t+(|S|-1)\epsilon$, where $\epsilon$ is arbitrarily small.

Fix any sequence $t_0 < t_1 < \ldots < t_n$ of $n+1$ points in time.
We first describe the gadgets we construct for the rows of matrix $M$.
At time $t_1$, we insert an edge between $u_j$ and $b$ for every $m_{1j} = 1$.
That is, we construct a retroactive operation $\create(\ins(u_j,b),t_1)$ for every $j\in[n]$ with $m_{1j} = 1$, resulting in $|m_1|$ retroactive operations at time (very close to) $t_1$.
Then for $i=2,\ldots,n$, at time $t_i$, we create $|m_{i-1}|+|m_i|$ retroactive operations at time $t_i$ as follows.
We delete all edges incident to $b$ (by operations $\create(\del(u_j,b),t_i)$ for all $j\in[n]$ with $m_{i-1,j}=1$), and create insertions of edges $(u_j,b)$ for every $j\in[n]$ with $m_{ij} = 1$ (by operations $\create(\ins(u_j,b),t_i)$ for all $j\in[n]$ with $m_{i j}=1$).
Our construction of the graph and retroactive operations guarantee that at time $t\in (t_i,t_{i+1}]$, $b$ is connected to $u_j$ if and only if $m_{ij} = 1$.
Next we describe the gadgets for the vectors $v_1,v_2,\ldots,v_n$.

At time $t_0$, we create an insertion of edge $(a,u_j)$ for every $j\in[n]$ with $v_{1j} = 1$.
Observe that $\query(a,b,t) = 1$ for $t\in ( t_i,t_{i+1} ]$ if and only if there exist some $u_j$ that is connected to both $a$ and $b$ at time $t$.
By the above construction, that implies $m_i\cdot v_1 = 1$.
Hence $n$ connectivity queries, namely at $t_1,t_2,\ldots,t_n$, between $a$ and $b$ suffice to compute $Mv_1$.
Given $v_2$, we modify the edges incident to $a$ as follows.
At time $t_0$, we delete all edges incident to $a$, and insert edge $(a,u_j)$ for every $j\in[n]$ with $v_{2j} = 1$ (with $O(n)$ retroactive operations).

In other words, we change the edges between $a$ and $\{u_j\}_{j\in[n]}$ at time $t_0$ based on $v_{2}$.
Then we can compute $M v_{2}$ by another $n$ connectivity queries as discussed above.
By repeating the above procedure for all vectors $v_i$, we can solve the OMv problem with $O(n^2)$ retroactive operations and queries, on a data structure with $O(n)$ nodes.
Hence if there exists a fully RDS for the connectivity problem with $O(n^{1-\epsilon})$ update and query time, then the OMv problem can be solved in $O(n^{3-\epsilon})$ time, violating the OMv conjecture.

\section{Incremental Fully Retroactive Connectivity and SF}\label{sec:incremental}

In this section we propose an incremental fully RDS for connectivity and spanning forest with polylogarithmic update and query time.
Recall that the edges are unweighted.
We first present the data structure to support connectivity queries.

Formally, an incremental fully retroactive connectivity data structure supports the following retroactive operations:
\begin{compactitem}
	\item $\create(\ins(e),t)$: insert an edge $e$ into the graph at time $t$;
	\item $\cancel(t)$: cancel the insertion of edge at time $t$; and
	\item $\query(u,v,t)$: return whether $u$ and $v$ are connected at time $t$.
\end{compactitem}

\begin{theorem}\label{th:incremental_conn}
	There exists an incremental fully retroactive connectivity data structure with amortized $O(\frac{\log^4 n}{\log\log n})$ update time that answers each query with worst case $O(\log n)$ time.
\end{theorem}
\begin{proof}
	Recall that the set $S$ (of updates) contains only insertions (each of them corresponds to a unique edge), while $\create()$ and $\cancel()$ modify $S$.
	Thus we can regard $S$ as a dynamic set of edges, where each edge has weight equal to the time it is inserted.
	The set $S$ defines an edge-weighted graph $H$, and the graph at time $t$ is the subgraph induced by edges with weight at most $t$.
	It suffices to maintain a dynamic MSF on the graph $H$: each $\create()$ inserts a weighted edge to $H$ and each $\cancel()$ deletes one from $H$. 
	
	We maintain a MSF on $H$ using the algorithm by Holm et al.~\cite{esa/HolmRW15}, and store the resulting MSF in a link-cut tree~\cite{jcss/SleatorT83}.
	Given the MSF, we can answer $\query(u,v,t)$ by looking at the edge with maximum weight $t'$ on the path between $u$ and $v$ in the MSF, and answer ``yes'' iff $t' < t$, which can be done in $O(\log n)$ time.
	It is not difficult to show the correctness of the query.
	Suppose there exists a path connecting $u$ and $v$ using edges of weight at most $t$ in $H$, then in the MSF, the maximum weight of an edge on the path between $u$ and $v$ must be at most $t$.
	Because otherwise we can remove that edge and include an edge with weight at most $t$, which violates the definition of MSF.
	
	Obviously, every retroactive operation and query can be handled by a single update on the MSF, which can be done in amortized $O(\frac{\log^4 n}{\log\log n})$ time.	
\end{proof}

Next we describe the data structure and algorithm to maintain an incremental fully retroactive SF.
To distinguish the SF from the MSF of $H$, we use $\text{MSF}_H$ to denote the weighted spanning forest of $H$ that we maintain.
We use the same data structure (with minor changes) to support the following queries:
\begin{compactitem}
	\item $\query(t)$: return a SF at time $t$;
	\item $\query(\text{size},t)$: return the size (number of edges) of a SF at time $t$.
\end{compactitem} 

Again, we maintain $\text{MSF}_H$ on $H$: $\query(t)$ can be trivially answered in $O(n)$ time by outputting all edges in the MSF$_H$ with weight less than $t$.
To support $\query(\text{size},t)$, we need to count the number of edges with weight less than $t$ in $\text{MSF}_H$.
We maintain an AVL tree that supports range query\footnote{Please refer to \url{https://www.geeksforgeeks.org/count-greater-nodes-in-avl-tree/} for an implementation.} on the weights of the edges of $\text{MSF}_H$.
Since every retroactive operation changes $\text{MSF}_H$ by at most one edge, the AVL tree can be maintained in $O(\log n)$ time per operation.
We can answer $\query(\text{size},t)$ by querying the number of elements with value less than $t$ in the AVL tree.
In summary, we have the following.

\begin{theorem}\label{th:incremental_msf}
	There exists an incremental fully retroactive SF with amortized $O(\frac{\log^4 n}{\log\log n})$ update time that supports $\query(t)$ in worst case $O(n)$ time and $\query(\text{size},t)$ in worst case $O(\log n)$ time.
\end{theorem}

While our data structure supports only unweighted edge insertions and deletions, we show that it can be extended to the weighted case to maintain an $(1+\epsilon)$-approximate MSF.
Using the techniques from Henzinger and King~\cite{sicomp/HenzingerK01}, we maintain an $(1+\epsilon)$-approximate MSF by partitioning the edges into weight classes.
Basically, we round the edge weights up to powers of $1+\epsilon$, and maintain $O(\frac{1}{\epsilon}\log W)$ incremental fully RDS we described above, one for each weight class.
Here we assume all edge weights are in $[1,W]$.
Each insertion of a weighted edge translates into an insertion of an unweighted edge in the corresponding weight class.
Queries for the approximation MSF made at time $t$ can be answered by collecting $O(\frac{1}{\epsilon}\log W)$ spanning forests (one from each data structure), and performing a static MSF algorithm, which takes time $O(\frac{n}{\epsilon}\log W)$.

In order to answer the total weight of the MSF more efficiently, we modify the data structure as follows.
Each insertion of an edge of weight $(1+\epsilon)^i$ is translated to an insertion of an unweighted edge in each of the weight classes $j = i, i+1,\ldots,l$, where $l=\log_{1+\epsilon}W$.
In other words, weight class $j$ contains all edges of weight \emph{at most} $(1+\epsilon)^j$.
Then the query of the total weight at time $t$ can be answered by $O(\frac{1}{\epsilon}\log W)$ queries $\query(\text{size},t)$ as follows.
Let $a_i$ be the size returned by $\query(\text{size},t)$ at weight class $i$, where $i=0,1,\ldots,l$.
Then $a_0 + \sum_{i=1}^{l}  (a_i - a_{i-1})\cdot(1+\epsilon)^i$ is the total weight of an $(1+\epsilon)$-approximation MSF.
Note that the query for the approximation MSF can still be answered by collecting $O(\frac{1}{\epsilon}\log W)$ spanning forests and performing a static MSF algorithm in $O(\frac{n}{\epsilon}\log W)$ time.
In summary, the amortized update time is $O(\frac{\log^4 n}{\log\log n}\cdot \frac{1}{\epsilon}\log W)$, and the worst case query time is $O(\frac{n}{\epsilon}\log W)$ for the approximation MSF, $O(\log n\cdot \frac{1}{\epsilon}\log W)$ for its total weight.

\section{Fully Retroactive Data Structures}\label{sec:fully}

In this section we present fully RDS for maintaining the maximum degree, connectivity and MSF.
Recall that for maintaining the maximum degree and MSF, edges are weighted.
Combined with the hardness results, the data structures we propose in this section achieve almost optimal (up to a polylogarithmic factor) time per operation. 
We first introduce a general framework for the fully RDS.


We present a dynamic balanced binary tree $\mathcal{T}$ that maintains the set of edges subject to insertions and deletions at different times.
The balanced binary tree serves as the framework for several RDS we will introduce later.
Depending on the problem, we maintain different (non-retroactive) dynamic data structures in the internal nodes.
We implement the balanced binary tree using the scapegoat tree~\cite{soda/GalperinR93}, which rarely rebuilds part of the tree to maintain balance.\footnote{Other balanced search trees, e.g., AVL tree~\cite{aw/Sedgewick83}, maintain balance by rotating part of the tree, which will be expensive when we maintain a data structure in each internal node $u$ depending on the set of leaves in $\mathcal{T}(u)$.}

We show that the balanced binary tree $\mathcal{T}$ enables us to handle each retroactive operation by updating $O(\log T)$ internal nodes if no rebuild occurs.
We rebuild the tree when it is not balanced and charge the cost of rebuild to the retroactive operations that are responsible for the imbalance, such that each operation is charged by $O(\log^2 T)$ updates of internal nodes.

Consider a sequence $S$ of $T$ updates and each update is associated with a time $t$.
We order the updates in $S$ in ascending order of their time, and we use $t_1<t_2<\ldots<t_T$ to denote these times.
For completeness, let $t_0 = -\infty$ and $t_{T+1} = \now$.
The scapegoat tree $\mathcal{T}$ we maintain has $T$ leaf nodes $( t_{i},t_{i+1} ]$ for $i = 1,2,\ldots,T$.
%
For any node $u$, let $\mathcal{T}(u)$ denote the subtree rooted at $u$ in $\mathcal{T}$.
The scapegoat tree maintains the following invariant:

\begin{invariant}\label{inv:balanced}
	For each internal node $u$ and its sibling $v$, $|\mathcal{T}(u)|\leq 2\cdot |\mathcal{T}(v)|$.
\end{invariant}

Whenever an internal node violates the invariant, the algorithm determines the internal node closest to the root that violates the invariant and rebuilds its subtree from scratch, fulfilling the invariant. The amortized cost of this rebuild is $O(\log T)$ per operation in $\mathcal{T}$.

A standard argument for balanced search tree implies that if the invariant is maintained, then the height of the tree is upper bounded by $O(\log T)$.
We maintain the following data structures for each node $u$ of the scapegoat tree:
\begin{itemize}
	\item an interval $I_u$, which is the union of the intervals of the leaves of $\mathcal{T}(u)$.
	\item a data structure $\mathcal{D}(u)$ that stores the edges $e$ such that (1) $I_u\subseteq L_e$; and (2) $I_w\nsubseteq L_e$, where $w$ is the parent of $u$ in $\mathcal{T}$.
	(Recall that $L_e$ is the lifespan of edge $e$.)
	If $u$ is the root of the tree then we only require that $I_u\subseteq L_e$.
	For convenience we also interpret $\mathcal{D}(u)$ as a set of edges.
	The exact choice of $\mathcal{D}(u)$ depends on the graph property that is maintained.
\end{itemize}

In other words, each internal node $u$ maintains an interval $I_u$ the subtree $\mathcal{T}(u)$ covers, and stores edge $e$ if the interval of $u$ is the maximal interval contained in $L_e$.
The above data structure enables efficient retrieval of $E_t$, i.e., the set of edges existing at time $t$.

\begin{lemma}\label{lemma:bst_coverage}
	Fix any time $t\in (t_i, t_{i+1}]$.
	Let $(v_l,v_{l-1},\ldots,v_0)$ be the path from the leaf node $v_l = ( t_i,t_{i+1} ]$ to the root $v_0$.
	We have $E_t = \bigcup_{i=0}^l \mathcal{D}(v_i)$ and $\mathcal{D}(v_i)\cap \mathcal{D}(v_j)=\emptyset$ for all $i\neq j$.
\end{lemma}
\begin{proof}
	First, for every $e\in E_t$ that exists at time $t$, we have $t\in L_e$, which implies that $v_l = ( t_i,t_{i+1} ]\subseteq L_e$.
	Thus $e$ must be contained in some unique $\mathcal{D}(v_i)$.
	That is, $E_t\subseteq  \bigcup_{i=0}^l \mathcal{D}(v_i)$.
	Specifically, $e$ is contained in $\mathcal{D}(v_i)$ such that $I_{v_i}\subseteq L_e$ while $I_{v_{i-1}}\nsubseteq L_e$.
	Therefore the sets of edges $\mathcal{D}(v_0),\mathcal{D}(v_1),\ldots,\mathcal{D}(v_l)$ are disjoint.
	On the other hand, for any $e\in \mathcal{D}(v_i)$, we have $I_{v_i}\subseteq L_e$, which implies $t\in (t_i, t_{i+1}] \subseteq L_e$ and hence $e\in E_t$.
\end{proof}

Lemma~\ref{lemma:bst_coverage} implies that with the tree $\mathcal{T}$, we can retrieve the edges $E_t$ by looking at $O(\log T)$ internal nodes.
In particular, $\query(t)$ can be handled by data structures maintained by $O(\log T)$ nodes.
For problems that admit linear time algorithms, e.g., connectivity and maximal matching, $\query(t)$ can be handled in $O(\log T + |E_t|)$ time, by maintaining the set of edges $\mathcal{D}(u)$ in each internal node $u$.
Next we show that the data structure maintains $O(\log T)$ copies of every edge.
Consequently, the total size of the sets $\mathcal{D}(u)$ is bounded by $O(T\log T)$.

\begin{lemma}\label{lemma:log_T_duplicates}
	Each edge is contained in $O(\log T)$ internal nodes.
	Moreover, these internal nodes can be found in $O(\log T)$ time.
\end{lemma}
\begin{proof}
	Fix any edge $e$ with $L_e = (t_a, t_b]$.
	By definition, if $\mathcal{D}(u)$ contains $e$ for some internal node $u$, then $I_u\subseteq L_e$ and $I_w\nsubseteq L_e$.
	Thus $w$ must be an ancestor of the leaf node $(t_{a-1}, t_{a}]$ or $(t_b, t_{b+1}]$, i.e., $I_w$ intersects with $L_e$ but is not contained in $L_e$.	
	Therefore, every internal node $u$ that contains $e$ must be a child of some node on the path from $(t_{a-1},t_{a}]$ to the root, or child of some node on the path from $(t_{b},t_{b+1}]$ to the root.
	Since the height of tree of $O(\log T)$ and each internal node has two children, there are $O(\log T)$ internal nodes containing $e$ and they can be found in $O(\log T)$ time.
\end{proof}

Next we show how to handle retroactive operations by updating the tree $\mathcal{T}$.
Intuitively, since each retroactive operation changes the lifespan of a single edge, by Lemma~\ref{lemma:log_T_duplicates}, the operation can be handled by updating $O(\log T)$ internal nodes.
However, to maintain a balanced binary tree, sometimes we need to rebuild part of the tree, which increases the amortized update time.

\begin{lemma}\label{lemma:update}
	Let $t_\text{update}$ be the update time of the data structure maintained in an internal node. 
	Each retroactive operation can be handled in amortized $O(\log^2 T\cdot t_\text{update})$ time.
\end{lemma}

Due to space limit, we defer the proof of the above lemma to the full version of the paper~\cite{henzinger2019upper}, where we give data structures maintaining maximum degree, connectivity and MSF subject to retroactive operations.
The data structures follow the above framework, while for different problems the data structures maintained by internal nodes are different.

\section*{Acknowledgment}
The research leading to these results has received funding from the European Research Council under the European Community's Seventh Framework Programme (FP7/2007-2013) / ERC grant agreement No. 340506.
Monika Henzinger acknowledges the Austrian Science Fund (FWF) and netIDEE SCIENCE project P 33775-N.

Xiaowei Wu is funded by the Science and Technology Development Fund, Macau SAR (File no. SKL-IOTSC-2021-2023), the Start-up Research Grant of University of Macau (File no. SRG2020-00020-IOTSC).

{
	\bibliography{retro}

\begin{thebibliography}{10}
\providecommand{\url}[1]{\texttt{#1}}
\providecommand{\urlprefix}{URL }
\providecommand{\doi}[1]{https://doi.org/#1}

\bibitem{de2020fully}
de~Andrade~J{\'u}nior, J.W., Duarte~Seabra, R.: Fully retroactive minimum
  spanning tree problem. The Computer Journal  (2020)

\bibitem{sicomp/BaswanaGS18}
Baswana, S., Gupta, M., Sen, S.: Fully dynamic maximal matching in o(log n)
  update time (corrected version). {SIAM} J. Comput.  \textbf{47}(3),  617--650
  (2018)

\bibitem{bentley1977algorithms}
Bentley, J.L.: Algorithms for klee's rectangle problems. Tech. rep., Technical
  Report, Computer (1977)

\bibitem{soda/BernsteinFH19}
Bernstein, A., Forster, S., Henzinger, M.: A deamortization approach for
  dynamic spanner and dynamic maximal matching. In: {SODA}. pp. 1899--1918.
  {SIAM} (2019)

\bibitem{soda/Blelloch08}
Blelloch, G.E.: Space-efficient dynamic orthogonal point location, segment
  intersection, and range reporting. In: {SODA}. pp. 894--903. {SIAM} (2008)

\bibitem{swat/ChenDGWXY18}
Chen, L., Demaine, E.D., Gu, Y., Williams, V.V., Xu, Y., Yu, Y.: Nearly optimal
  separation between partially and fully retroactive data structures. In:
  {SWAT}. LIPIcs, vol.~101, pp. 33:1--33:12. Schloss Dagstuhl - Leibniz-Zentrum
  fuer Informatik (2018)

\bibitem{talg/DemaineIL07}
Demaine, E.D., Iacono, J., Langerman, S.: Retroactive data structures. {ACM}
  Trans. Algorithms  \textbf{3}(2), ~13 (2007)

\bibitem{wads/DemaineKLSY15}
Demaine, E.D., Kaler, T., Liu, Q.C., Sidford, A., Yedidia, A.: Polylogarithmic
  fully retroactive priority queues via hierarchical checkpointing. In: {WADS}.
  Lecture Notes in Computer Science, vol.~9214, pp. 263--275. Springer (2015)

\bibitem{esa/DickersonEG10}
Dickerson, M.T., Eppstein, D., Goodrich, M.T.: Cloning voronoi diagrams via
  retroactive data structures. In: {ESA} {(1)}. Lecture Notes in Computer
  Science, vol.~6346, pp. 362--373. Springer (2010)

\bibitem{jcss/DriscollSST89}
Driscoll, J.R., Sarnak, N., Sleator, D.D., Tarjan, R.E.: Making data structures
  persistent. J. Comput. Syst. Sci.  \textbf{38}(1),  86--124 (1989)

\bibitem{jal/FiatK03}
Fiat, A., Kaplan, H.: Making data structures confluently persistent. J.
  Algorithms  \textbf{48}(1),  16--58 (2003)

\bibitem{soda/GalperinR93}
Galperin, I., Rivest, R.L.: Scapegoat trees. In: {SODA}. pp. 165--174.
  {ACM/SIAM} (1993)

\bibitem{talg/GiyoraK09}
Giyora, Y., Kaplan, H.: Optimal dynamic vertical ray shooting in rectilinear
  planar subdivisions. {ACM} Trans. Algorithms  \textbf{5}(3),  28:1--28:51
  (2009)

\bibitem{isaac/GoodrichS11}
Goodrich, M.T., Simons, J.A.: Fully retroactive approximate range and nearest
  neighbor searching. In: {ISAAC}. Lecture Notes in Computer Science,
  vol.~7074, pp. 292--301. Springer (2011)

\bibitem{stoc/HenzingerKNS15}
Henzinger, M., Krinninger, S., Nanongkai, D., Saranurak, T.: Unifying and
  strengthening hardness for dynamic problems via the online matrix-vector
  multiplication conjecture. In: {STOC}. pp. 21--30. {ACM} (2015)

\bibitem{henzinger2019upper}
Henzinger, M., Wu, X.: Upper and lower bounds for fully retroactive graph
  problems. arXiv preprint arXiv:1910.03332  (2019)

\bibitem{jacm/HenzingerK99}
Henzinger, M.R., King, V.: Randomized fully dynamic graph algorithms with
  polylogarithmic time per operation. J. {ACM}  \textbf{46}(4),  502--516
  (1999)

\bibitem{sicomp/HenzingerK01}
Henzinger, M.R., King, V.: Maintaining minimum spanning forests in dynamic
  graphs. {SIAM} J. Comput.  \textbf{31}(2),  364--374 (2001)

\bibitem{jacm/HolmLT01}
Holm, J., de~Lichtenberg, K., Thorup, M.: Poly-logarithmic deterministic
  fully-dynamic algorithms for connectivity, minimum spanning tree, 2-edge, and
  biconnectivity. J. {ACM}  \textbf{48}(4),  723--760 (2001)

\bibitem{esa/HolmRW15}
Holm, J., Rotenberg, E., Wulff{-}Nilsen, C.: Faster fully-dynamic minimum
  spanning forest. In: {ESA}. Lecture Notes in Computer Science, vol.~9294, pp.
  742--753. Springer (2015)

\bibitem{thesis/parsa2014algorithms}
Parsa, S.: Algorithms for the Reeb Graph and Related Concepts. Ph.D. thesis,
  Duke University (2014)

\bibitem{sicomp/PatrascuD06}
Patrascu, M., Demaine, E.D.: Logarithmic lower bounds in the cell-probe model.
  {SIAM} J. Comput.  \textbf{35}(4),  932--963 (2006)

\bibitem{sicomp/RodittyZ16}
Roditty, L., Zwick, U.: A fully dynamic reachability algorithm for directed
  graphs with an almost linear update time. {SIAM} J. Comput.  \textbf{45}(3),
  712--733 (2016)

\bibitem{aw/Sedgewick83}
Sedgewick, R.: Algorithms. Addison-Wesley (1983)

\bibitem{jcss/SleatorT83}
Sleator, D.D., Tarjan, R.E.: A data structure for dynamic trees. J. Comput.
  Syst. Sci.  \textbf{26}(3),  362--391 (1983)

\bibitem{focs/Solomon16}
Solomon, S.: Fully dynamic maximal matching in constant update time. In:
  {FOCS}. pp. 325--334. {IEEE} Computer Society (2016)

\bibitem{jacm/Tarjan75}
Tarjan, R.E.: Efficiency of a good but not linear set union algorithm. J. {ACM}
   \textbf{22}(2),  215--225 (1975)

\end{thebibliography}
	\bibliographystyle{splncs04}
}

\end{document}